\documentclass{amsart}
\usepackage{amsmath,amsfonts, amssymb,amsthm}
\usepackage{graphicx}
\usepackage{dcolumn}
\usepackage{bm}
\usepackage[backend=biber,style=alphabetic]{biblatex}

\newcommand{\N}{\mathbb{N}}
\newcommand{\Z}{\mathbb{Z}}
\newcommand{\R}{\mathbb{R}}
\newcommand{\C}{\mathbb{C}}

\renewcommand{\P}{\mathbb{P}}
\newcommand{\E}{\mathbb{E}}

\newcommand{\ve}{\varepsilon}

\newtheorem{thm}{Theorem}[section]
\newtheorem{thmcite}[thm]{Theorem}
\newtheorem{lem}[thm]{Lemma}
\newtheorem{cor}[thm]{Corollary}
\newtheorem{remm}[thm]{Remark}
\newtheorem{deffo}[thm]{Definition}
\newtheorem{prop}[thm]{Proposition}
\begin{document}

	\title[Lifting method for exponential large deviation estimates]{A ``lifting'' method for exponential large deviation estimates and an application to certain non-stationary 1D lattice
		Anderson models}
	\author{Omar Hurtado}
	
	\address{University of California, Irvine\\
		ohurtad1@uci.edu}

	\date{\today}
	
	\begin{abstract}
		Proofs of localization for random Schr\"odinger operators with sufficiently regular distribution of the potential can take advantage of the fractional moment method introduced by Aizenman-Molchanov, or use the classical Wegner estimate as part of another method, e.g. the multi-scale analysis introduced by Fr\"ohlich-Spencer and significantly developed by Klein and his collaborators. When the potential distribution is singular, most proofs rely crucially on exponential estimates of events corresponding to finite truncations of the operator in question; these estimates in some sense substitute for the classical Wegner estimate. We introduce a method to ``lift'' such estimates, which have been obtained for many stationary models, to certain closely related non-stationary models. As an application, we use this method to derive Anderson localization on the 1-D lattice for certain non-stationary potentials along the lines of the non-perturbative approach developed by Jitomirskaya-Zhu in 2019.
	\end{abstract}
	
	\maketitle
	
	\section{Introduction}
	The study of large deviations has featured prominently in many proofs of localization for ergodic Schr\"odinger operators, both in the random and quasi-periodic context. For the random case, the first proof of Anderson localization for the Anderson model with singular potentials on the one-dimensional lattice in \cite{Carmona1987} by Carmona-Klein-Martinelli used large deviation estimates coming from a study of the Lyapunov exponent; later proofs in e.g. \cite{Shubin1998}, \cite{Bucaj2019LocalizationFT}, \cite{GORODETSKI}, and \cite{SJZhu19} either introduced new approaches or made simplifications, but still relied on large deviation estimates in some way.
	
	For the two-dimensional Anderson model with singular potentials, Ding-Smart in \cite{Ding-Smart} recently obtained localization at the edges of the spectrum, and a key unique continuation result relied crucially on large deviation estimates in the form of the Azuma inequality.   This unique continuation enabled them to use the strategy introduced by Bourgain-Kenig in \cite{Bourgain2005} and elucidated by Germinet-Klein in \cite{Germinet2001}.  The result in \cite{Ding-Smart} has since been improved upon in \cite{Li20}, and was combined with a $\Z^3$ unique continuation result to obtain the same localization result in the three-dimensional context in \cite{Li-Zhang}. In the quasi-periodic context, analysis of the large deviation sets is crucial for e.g. the method introduced by Bourgain and his co-authors in \cite{BG00}, detailed in \cite{BourgainBook}. Recent work in this vein includes \cite{Jitomirskaya2020} and \cite{LiuAPDE}.
	
	To be more explicit, when we speak of discrete Schr\"odinger operators in dimension $d$ we mean operators on $\ell^2(\Z^d)$ of the form
	\begin{equation}\label{schrodgen} [H \psi](n) = \sum_{|m-n|=1} \psi(m) + V_n\psi(n)
	\end{equation}
	In both the random and quasi-periodic cases, when large deviation estimates are under consideration, they are often of the form
	\begin{equation}\label{largedev}\P[\text{A square of length } \ell\text{ is ``bad''}] \leq e^{-\eta \ell}
	\end{equation} for some $\eta > 0$ and some contextual notion of ``bad''. In general, most proofs of localization require statements like the above, but with $e^{-\eta \ell}$ possibly replaced by some other $f(\ell)$ vanishing as $\ell \rightarrow \infty$. For the special case in (\ref{largedev}) of an exponential decay, we've produced a method by which to ``lift'' estimates of this form from some joint distribution $\P_1$ of the values $\{V_n(\omega)\}$ to another distribution $\P_0$ which is closely related. We note that our method admits some generalization to arbitrary $f(\ell)$, but both the constraints on $\P_0$ and the losses in the ``lifted'' rate can be severe for slower rates of decay, motivating our emphasis on exponential estimates. 
	
	In particular, the method is still somewhat robust for ``almost exponential'' decays like $f(\ell) = e^{-\eta\ell^q}$ (where $0< q < 1$), and not very robust for power law decays, i.e. $f(\ell) = \ell^{-q}$ with $q > 0$. The method requires estimates on how ``close'' $\P_0$ is to $\P_1$ when both are restricted to finitely many coordinates, so we expect applications to be mostly (or exclusively) in the random context, where the potential at any given site has no dependence on the potential at other sites; in such cases, said estimates can be obtained in fairly natural circumstances. We note however that the underlying probabilistic results apply in all contexts where the requisite estimates hold.
	
	Specifically, we let $\mu$ be a non-trivial probability distribution on $\R$ satisfying an appropriate moment condition. (By non-trivial, we mean $\mu$ is supported on at least two points.) When $d=1$ and $V_n$ are identically distributed with law $\mu$, the transfer matrices associated to the random operator
	\begin{equation*}[H \psi](n) = \psi(n+1) + \psi(n-1) + V_n\psi(n) \end{equation*}
	satisfy large deviation estimates as a consequence of the work in \cite{LePage}, which should be understood in the wider context of the Furstenberg theory of random matrix products. More explicitly, we let $\Omega = \R^\Z$, $\P = \mu^\Z$, and $V_n:\Omega\rightarrow \R$ be projection onto the $n$-th coordinate. Then the family of operators
	\begin{equation} \label{schrod}
		[H_\omega\psi](n) = \psi(n+1)+\psi(n-1)+V_n(\omega)\psi(n)
	\end{equation}
	are precisely distributed in this way. This explicit description is necessary both in comparing such families to families with slightly different distribution, and in the method used by \cite{SJZhu19} specifically which we hope to adapt to our context.
	
	For a system where $V_n$ are independently distributed with laws of the form $g_n\mu$, with $g_n > 0$, $\int g_nd\mu = 1$ and $g_n \in L^\infty(d\mu)$, we derive the existence of similar large deviation estimates for the transfer matrices associated to said system so long as the key condition below holds: 
	\begin{equation}\label{logmom} \lim_{N\rightarrow \infty} \frac{1}{N} \sum_{n=-N}^N \log \|g_n\|_\infty = 0 
	\end{equation}
	i.e. $\sum_{n=-N}^N \log \|g_n\|_\infty = o(N)$. Further, this system will show more or less the same Lyapunov behavior as if $V_n(\omega)$ were i.i.d. with law $\mu$. By this we mean that the ``approximate'' system with potential at site $n$ distributed with law $g_n\mu$ has a Lyapunov exponent $\tilde{\gamma}(E)$ which describes the asymptotics of the associated transfer matrices, and at all energies this coincides with the Lyapunov exponent $\gamma(E)$ describing the transfer matrices associated to the ``exact'' system where the potential at each site is distributed with law $\mu$. This probabilistic result can loosely be understood as a multiplicative version of the standard result found in e.g. \cite{dupuis_ellis} that large deviation principles for stochastic processes are preserved under certain forms of super-exponential approximation. 
	
	These estimates are then used to derive almost sure Anderson localization for the associated random family of Schr\"odinger operators, using a strategy for proving localization along the lines of \cite{SJZhu19} and \cite{Rangamani19}. In the case where $\mu$ has unbounded support, we will require further technical assumptions, the most salient of which is the strengthening of (\ref{logmom}) to
	\begin{equation}\label{logmomunif}
		\lim_{N\rightarrow \infty} \frac{1}{N} \left[\sup_{k\in\Z}\sum_{n=k-N}^{k+N}\log\|g_n\|_\infty \right] = 0
	\end{equation}

	To our knowledge, the only work which has obtained large deviation estimates for non-stationary matrix products is from Gorodetski-Kleptsyn in \cite{GORODETSKI2}; \cite{Goldsheid2022} also analyzes the asymptotics of non-stationary random matrix products, but does not produce large deviation estimates. We note the the large deviation estimates found in \cite{GORODETSKI2} hold in an overall more general setting than ours, forgoing any kind of convergence condition on the sequence of distributions and the requirement that they all be absolutely continuous with respect to some base distribution $\mu$. 
	
	However, in this considerably more general setting, a Lyapunov exponent cannot be expected to exist for all the non-stationary matrix products satisfying the relevant assumptions. In particular, to our knowledge the results in \cite{GORODETSKI2} work cannot be combined with the Craig-Simon subharmonicity result in \cite{CS83}, which uses subharmonicity of the Lyapunov exponent $\gamma(z)$ to achieve a certain kind of uniformity in energy.
	
	The general large deviations they obtain consequently do not seem to be compatible with the approach of \cite{SJZhu19}, which we use to derive localization. We came to know while preparing this work that forthcoming work of Gorodetski-Kleptsyn will prove Anderson localization in the context of distributions satisfying the requirements in \cite{GORODETSKI2} via completely different methods, using a purely dynamical approach which builds upon the approach used in \cite{GORODETSKI}.
	
	We emphasize that the results in \cite{GORODETSKI2} are on the whole considerably more general in the context of large deviation estimates corresponding to non-stationary matrix products. Similarly, the forthcoming localization result from Gorodetski-Kleptsyn will be far more general then the localization results we obtain here. There are certain rare examples for which our results are applicable and those in \cite{GORODETSKI2} are not, but ``most'' examples satisfying the hypotheses of our work satisfy the hypotheses in \cite{GORODETSKI2}. 
	
	At the same time, the method presented there is one-dimensional in an essential way, as the transfer matrix method which allows the study of generalized eigenfunctions in terms of matrix products is only available in one-dimensional and quasi-one-dimensional contexts. In contrast, our lifting method, while imposing much harsher restrictions on the distributions considered, is fundamentally an abstract probabilistic result about product measures, and is easily adapted to the study of non-stationary potentials on $\Z^d$ for any $d$, or any reasonably ``tame'' lattice. 
	
	Moreover, while the large deviation work in \cite{GORODETSKI2} does not require compact support of the distribution and makes morally the same mild moment assumptions as we do in this work and as are made in e.g. \cite{Carmona1987}, \cite{Rangamani19}, it is our understanding that the forthcoming localization proof from Gorodetski-Kleptsyn will further the approach of \cite{GORODETSKI} and in particular will require the potentials $V_n(\omega)$ to have a uniform bound. We point out that our localization result holds for unbounded distributions under our more stringent assumptions.
	
	These results are of special interest in the case where $\mu$ is singular. With sufficient regularity, the study of the problem becomes amenable to the fractional moment method pioneered in \cite{Aizenman1993}; such methods are able to handle non identically distributed distributions in considerably higher generality than our methods under the requisite assumptions. The methods of \cite{ks81}, developed earlier, are also available for the one-dimensional problem specifically. In the continuum setting, Klein obtained localization for certain non-stationary Anderson models with no regularity assumptions made in \cite{Klein2013}. This was a consequence of Wegner estimates derived more generally using the quantitative unique continuation introduced by Bourgain-Kenig in \cite{Bourgain2005} and further elaborated upon by Bourgain-Klein in \cite{Bourgain2013}.
	
	In addition, the requirement that the laws of $V_n(\omega)$ are all absolutely continuous with respect to some base measure $\mu$, becomes considerably more natural in the context of finite valued potentials specifically. If e.g. $\mu = \sum_{m=1}^M \alpha_m \delta_{x_m}$ then absolute continuity implies that $\mu_n = \sum_{m=1}^M \beta_{n,m} \delta_{x_m}$, and $\|g_n\|_\infty= \max\{\frac{\beta_{n,m}}{\alpha_m}\}$. Physically, the situation corresponds to pockets of zero natural density where the probability differs from what it should be more than any given $\ve > 0$; (\ref{logmom}) is in this context precisely equivalent to $\beta_{n,m} \xrightarrow{d} \alpha_m$ for all $m$, where the convergence is in the sense of natural density, or more precisely its obvious analogue on $\Z$.
	
	We also prove another result under even stronger assumptions than (\ref{logmomunif}) that essentially makes all the statistics of a system with $V_n(\omega)$ distributed as $g_n\mu$ identical to those of i.i.d. $V_n(\omega)$ with law $\mu$; specifically if
	\begin{equation}\label{logsum}
		\sum_{n\in\Z} \log\|g_n\|_\infty < \infty
	\end{equation}
	then the joint distribution $\P_0$ corresponding to the non-stationary case is absolutely continuous with respect to the joint distribution $\P_1$ with $V_n(\omega)$ i.i.d. This argument, under more stringent conditions, essentially allows a great number of results, like e.g. the dynamical localization results from \cite{SJZhu19}, to be carried over wholesale from stationary contexts to non-stationary contexts.
	
	We mention in brief that our model is just one of many where the potential incorporates randomness but is not just given by i.i.d. random variables. A comprehensive review is well beyond the scope of this work; we briefly mention the random polymer models studied in e.g. \cite{Jitomirskaya2003} and \cite{Rangamani2022} which have potential purely driven by randomness but allow some ``local'' dependence among the variables determining the potential, mixed models studied in e.g. \cite{Cai2022}, \cite{Damanik2022} which consider potentials given by random terms together with terms which are either quasi-periodic  or periodic, and ``trimmed'' models considered in e.g. \cite{Elgar2014} and \cite{Elgar2017}.
	
	The rest of the paper is organized as follows: in Section \ref{prelim} we introduce the necessary probabilistic definitions and our results, introducing our three main probabilistic results and their consequences for localization of non-stationary Anderson models. In Section \ref{prob}, we prove our probabilistic results. In Section \ref{schrodd} we recall basic facts about Schr\"odinger operators and results regarding stationary Anderson models. In Section \ref{ldeland} we use our probabilistic results to derive large deviation estimates and prove important consequences thereof, namely identical Lyapunov behavior for the non-stationary approximate system and applicability of the Craig-Simon subharmonicity result for said system. In Section \ref{lemma4evr}, we prove technical lemmas necessary to prove our main localization result concerning unbounded potentials, Theorem \ref{speclocunb}, and comment on the small changes necessary to prove the similar Theorem \ref{specloc} which allows weaker hypotheses in the case that the potential is bounded. Finally, in Section \ref{fini4evr}, we prove Theorem \ref{speclocunb}.

	\section{Preliminaries and statements of results}\label{prelim}
	
	Throughout, $\P_0$ and $\P_1$ will denote distinct probability distributions on the same measure space $(X, \mathcal{F})$. When the discussion is specified to Schr\"odinger operators, $(X,\mathcal{F})$ will be $\R^\Z$ and the product Borel algebra. In this case, $\P_0$ and $\P_1$ can also be understood as joint distributions of variables $V_n(\omega)$, so that the $V_n$ are projection maps. The expectation with respect to $\P_0$ and $\P_1$ will be denoted $\E_0$ and $\E_1$ respectively. Whenever $\mathcal{G} \subset \mathcal{F}$ are $\sigma$-algebras and $X$ is an $L^1(\P_i,\mathcal{F})$ random variable, we denote its conditional expectation with respect to  $\mathcal{G}$ and $\P_i$ by $\E_i[X\,|\,\mathcal{G}]$. That is, $\E_i[X\,|\,\mathcal{G}]$ is the unique (up to $\P_i$-a.e. equivalence) $\mathcal{G}$ measurable variable which satisfies
	\[ \E_i[\chi_A\cdot X] = \E_i[\chi_A\cdot \E_i[X\,|\,\mathcal{G}]]\]
	for all $A \in \mathcal{G}$. Throughout, $\P_1$ can be considered an ``exact'' distribution, and $\P_0$ an ``approximate'' distribution. 
	
	For two measures $\nu$ and $\mu$ on a measure space, $\nu \ll \mu$ denotes absolute continuity of $\nu$ with respect to $\mu$. When we specify the conversation to product measures, we will fix a probability measure $\mu$ on $\R$, and consider non-negative $g_n \in L^\infty(d\mu)$ with $g \geq 0$ and $\int g d\mu = 1$. In a natural way these functions correspond to probability measures $\nu \ll \mu$ with essentially bounded Radon-Nikodym derivatives. In some sense these measures are the more relevant objects, but we identify them with their densities for notational simplicity, so that we may write e.g. $g_n$ and $g_n\mu$ instead of $\frac{d\mu_n}{d\mu}$ and $\mu_n$.
	
	Given $(\Omega, \mathcal{F})$ some measure space, $\mathcal{F}_n$ a filtration and $\P_0$, $\P_1$ two probability measures such that
	$\P_0|_{\mathcal{F}_n} \ll \P_1|_{\mathcal{F}_n}$, we will define
	\[ H_n:= \frac{d(\P_0|_{\mathcal{F}_n})}{d(\P_1|_{\mathcal{F}_n})}\]
	Sometimes we will consider families of filtrations, i.e. collections $(\mathcal{F}_n^k)$ where for any fixed $k \in \Z$ $(\mathcal{F}_n^{k_0})$ is a filtration. In this case, we set
	\[ H_n^k:= \frac{d(\P_0|_{\mathcal{F}_n^k})}{d(\P_1|_{\mathcal{F}_n^k})}\] We introduce notions relevant only to the study of Schr\"odinger operators later in the paper.
	
	\begin{deffo}
		We let $(X,\mathcal{M})$ be a measurable space, and $(A^E_n)_{(n,E) \in \N \times \R}$ a collection of measurable sets. (This should be understood as a collection of sequences indexed by $E \in \R$.) We say that the collection decays exponentially uniformly in $E$ with respect to a probability measure $\P$ if there exists $N \in \N$ and $\eta >0$ such that for $n> N$ and all $E$, we have
		\[ \P[A^E_n] \leq e^{-\eta n}\]

	\end{deffo}
	\begin{deffo}	
		We say a collection of sequences of events $(A_n^E)_{(n,E)}$ (with $n \in \N, E \in \R$) is adapted to a filtration of $\sigma$-algebras $(\mathcal{F}_n)_{n \in \N}$ if for all $E$ and $n$ we have $A_n^E \in \mathcal{F}_n$.
	\end{deffo}
	
	Note that there is nothing special about the choice of $\R$ for indexing our collection of sequences, save that it is what we will use later in this work, with $E$ representing energy.
	
	\begin{thm} \label{transmeth}
		Let $(X,\mathcal{M})$ be a measurable space. Further let $\Omega = X^\Z$ and $\mathcal{B}$ the $\sigma$-algebra on $X^\Z$ generated by measurable cylinder sets. Let $\mu$ be a probability measure on $(X,\mathcal{M})$ and $g_n$ a sequence of non-negative functions with $g_n \in L^\infty(d\mu)$ and $\int g_nd\mu = 1$. We define on $(\Omega,\mathcal{B})$ the product measures:
		\[ \P_0 = \bigotimes_{n\in\Z} g_n\mu,\quad \P_1 = \mu^\Z\]
		and the coordinate projections $V_n$ by $V_n(\omega) = \omega_n$ for $\Omega \ni \omega = (\omega_n)_{n\in\Z}$.
		Finally, we define the $\sigma$-algebras
		\[\mathcal{F}_n = \begin{cases} \sigma(V_{-n},V_{1-n},\ldots V_{n-1},V_n)\quad &n>0\\
			\{\Omega,\varnothing\}\quad &n=0 \end{cases}\]
		If the $g_n$ satisfy:
		\begin{equation}
			\lim_{N\rightarrow \infty} \frac{1}{N} \sum_{n=-N}^N\log\|g_n\|_\infty = 0 \tag{\ref{logmom}}
		\end{equation}
		then any collection of events $(A_n^E)$ which is adapted to $(\mathcal{F}_n)_{n \in \N}$ and exponentially decaying uniformly in $E$ in $\P_1$ is also exponentially decaying uniformly in $E$ in $\P_0$.
	\end{thm}
	
	A straightforward consequence of Theorem \ref{transmeth} is large deviation estimates valid for the joint distribution $\P_0$, and as a consequence the existence of a Lyapunov exponent, the same one existing for the joint distribution $\P_1$.
	\begin{thm}
		Let $\Omega = \R^\Z$ and $\P_1 = \mu^\Z$ for some non-trivial (i.e. supported on at least two points) $\mu$ such that there is $\alpha> 0$ for which $\int |x|^\alpha d\mu(x) < \infty$, and define the family of operators $H_\omega$ by
		\[ H_\omega \psi(n) = \psi(n+1) + \psi(n-1) + V_n(\omega)\psi(n)\]
		Moreover, let $S_{[1,n]}^z(\omega)$ be the $SL_2(\C)$ matrices satisfying
		\[ S_{[1,n]}^z(\omega)\begin{pmatrix} \psi(1) \\ \psi(0) \end{pmatrix}  = \begin{pmatrix} \psi(n+1) \\ \psi(n) \end{pmatrix}\]
		for any solution $\psi \in \C^\Z$ to $H_\omega \psi = z\psi$, and $\gamma(z) = \lim_{n\rightarrow \infty} \frac{1}{n} \E_1[\log\|S_{[1,n]}^z(\omega)\|]$. Then for any fixed $z \in \C$ and $\ve >0$, there exist $\eta = \eta(z,\ve) > 0$ and $N = N(z,\ve)$ such that
		\[\P_0\left[\left|\frac{1}{n}\log\|S_{[1,n]}^z(\omega)\| - \gamma(z)\right| > \ve\right] \leq e^{-\eta n}\]
		for $n>N$. In particular, $\frac{1}{n}\|S_{[1,n]}^z(\omega)\|\rightarrow \gamma(z)$ $\P_0$-almost surely.
	\end{thm}
	These results suffice to produce the necessary estimates to show Anderson localization for the bounded case:
	\begin{thm}\label{specloc}
		Let $(X,\mathcal{M})$ be $\R$ equipped with its Borel $\sigma$-algebra, and further let $(\Omega,\mathcal{B})$, $\P_0$, $\P_1$, $\mathcal{F}_n$ and $V_n$ be as in Theorem \ref{transmeth}, with equation (\ref{logmom}) satisfied.
		
		Assume moreover that $\mu$ has compact support. Then there is a set $\Omega_0\subset \Omega$ with $\P_0[\Omega_0] =1$ such that for all $\omega \in \Omega_0$, the operators on $\ell^2(\Z)$ defined by
		\[H_\omega\psi (n) = \psi(n-1) + \psi(n+1) + V_n(\omega)\psi(n)\]
		has pure point spectrum with all eigenfunctions exponentially decaying.
	\end{thm}
	However, for distributions $\mu$ with unbounded support, we need for technical reasons to be able to discuss uniformity across different choices of filtrations $(\mathcal{F}_n)^k$, which physically corresponds to uniformity across different choices of centers for an interval. Concretely, in the stationary Anderson model context, the statistics for the transfer matrix $S_{[a,b],\omega}^E$ only depend on the length $b-a+1$ and the energy $E$. While we cannot recover this exact statement in the non-stationary case, we can under stronger assumptions lift a priori large deviation estimates which are uniform across such choices from a stationary context to similar large deviation estimates in a non-stationary context. We introduce additional definitions specific to these considerations. Throughout, ``uniform in filtration'' can loosely be read as ``depending only on the length of the associated interval''. This is in the full level of abstraction not quite correct, but describes our specific application. We introduce necessary definitions:
	
	\begin{deffo}
		Given a measure space $(X,\mathcal{M})$ and a collection of filtrations $(\mathcal{F}_n)^k$, with $n \in \N, k \in \Z$, we say a collection of events $(A_n^{k,E})$ decays exponentially uniformly in $E$ and $k$ with respect to $\P$ if there exists $N$ and $\eta > 0$ such that for all $E$, all $k$, and all $n > N$, we have
		\[ \P[A_n^{k,E}] < e^{-\eta n}\]
	\end{deffo}
	\begin{deffo}
		We say $(A_n^{k,E})$ is $(\mathcal{F}_n)^k$ adapted if for any fixed $k_0 \in \Z$, the collection of events $A_n^{k_0,E}$ is $\mathcal{F}_n^{k_0}$ adapted.
	\end{deffo}
	With these natural extensions of earlier definitions, we can state a version of Theorem \ref{transmeth} which is uniform in filtration.
	
	\begin{thm} \label{transmethunif}
		Let $(X,\mathcal{M})$ be a measurable space. Further let $\Omega = X^\Z$ and $\mathcal{B}$ the $\sigma$-algebra on $X^\Z$ generated by measurable cylinder sets. Let $\mu$ be a probability measure on $(X,\mathcal{M})$ and $g_n$ a sequence of non-negative functions with $g_n \in L^\infty(d\mu)$ and $\int g_nd\mu = 1$. We define on $(\Omega,\mathcal{B})$ the product measures:
		\[ \P_0 = \bigotimes_{n\in\Z} g_n\mu,\quad \P_1 = \mu^\Z\]
		and the coordinate projections $V_n$ by $V_n(\omega) = \omega_n$ for $\Omega \ni \omega = (\omega_n)_{n\in\Z}$.
		Finally, we define the $\sigma$-algebras
		\[\mathcal{F}_n^k = \begin{cases} \sigma(V_{k-n},V_{k+1-n},\ldots V_{k+n-1},V_{k+n})\quad &n>0\\
			\{\Omega,\varnothing\}\quad &n=0 \end{cases}\]
		If the $g_n$ satisfy:
		\begin{equation}
			\lim_{N\rightarrow \infty} \frac{1}{N} \left[\sup_{k\in\Z}\sum_{n=k-N}^{k+N}\log\|g_n\|_\infty \right] = 0 \tag{\ref{logmomunif}}
		\end{equation}
		then any collection of events $(A_n^{E,k})$ which is adapted to $(\mathcal{F}_n)^k$ and exponentially decaying uniformly in $E$ and $k$ with respect to $\P_1$ is also exponentially decaying uniformly in $E$ and $k$ with respect to $\P_0$.
	\end{thm}

	This version, uniform in filtration (i.e. center), suffices to show Anderson localization for $\mu_n$ ``converging'' to $\mu$ which have unbounded support but satisfy a mild moment condition, more or less the condition found in the original Carmona-Klein-Martinelli work \cite{Carmona1987}.
	\begin{thm}\label{speclocunb}
		Let $(X,\mathcal{M})$ be $\R$ equipped with its Borel $\sigma$-algebra, and further let $(\Omega,\mathcal{B})$, $\P_0$, $\P_1$, $\mathcal{F}_n$ and $V_n$ be as in Theorem \ref{specloc} with equation (\ref{logmomunif}) satisfied.
		
		Assume also that there are $\alpha > 0$ and $M < \infty$ such that $\int |x|^\alpha d\mu(x) < K$ and $\int |x|^\alpha g_n(x)d\mu(x) < M$ for all $n \in \Z$.  set $\Omega_0\subset \Omega$ with $\P_0[\Omega_0] =1$ such that for all $\omega \in \Omega_0$, the operators on $\ell^2(\Z)$ defined by
		\[H_\omega\psi (n) = \psi(n-1) + \psi(n+1) + V_n(\omega)\psi(n)\]
		has pure point spectrum with all eigenfunctions exponentially decaying, i.e. $H_\omega$ exhibits Anderson localization.
	\end{thm}
	\begin{remm}
		Our additional conditions in the unbounded case amount to imposing uniformity in various ways; this in some sense is required to make up for the loss of uniformity which came from the existence of a bound on $V_n$. In particular (\ref{logmomunif}) allows the extraction of large deviation estimates uniform in ``center'' of the corresponding interval or square, by prohibiting arbitrarily long stretches of abnormally high $\|g_n\|$. We will explicitly go through the unbounded case through the rest of the paper, and point out when they arise the places where boundedness allows one to discard assumptions.
	\end{remm}
	Finally, a condition stronger than even (\ref{logmomunif}) forces essentially all the relevant statistics of $\P_0$ to coincide with those of $\P_1$, i.e. absolute continuity with an essential bound on $\frac{d\P_0}{d\P_1}$.
	\begin{thm}\label{abscont}
		Let $(\Omega,\mathcal{B})$, $\P_0$, and $\P_1$ be as in Theorem \ref{transmeth}. If the Radon-Nikodym derivatives satisfy the stronger condition
		\begin{equation}
			\sum_{n\in\Z} \log \|g_n\|_\infty < \infty \tag{\ref{logsum}}
		\end{equation}
		then $\P_0 \ll \P_1$, and moreover for any $A \in \mathcal{B}$, we have $\P_0[A] \leq C \P_1[A]$, where $C := \prod_{n\in\Z} \|g_n\|_\infty$.
	\end{thm}
	
	A direct application of this result to the results in \cite{SJZhu19} gives the following result:
	\begin{thm}\label{dynloc}
		Let everything be as in Theorem \ref{specloc}; assume further that $\mu$ has compact support and the Radon-Nikodym derivatives obey the stronger condition
		\begin{equation}
			\sum_{n\in\Z} \|g_n\|_\infty < \infty \tag{\ref{logsum}}
		\end{equation}
		Then $H_\omega$ is almost surely exponentially dynamically localized, in the sense of \cite{SJZhu19}.
	\end{thm}
	\section{Lifting Method, Probabilistic Results}\label{prob}
	We prove a simple lemma, and a useful corollary, before proving Theorem \ref{transmeth} which allows us to lift large deviation estimates under certain conditions.
	\begin{lem}\label{slowgrowlemunif}
		Let $(X,\mathcal{M})$ be a be a measurable space, and $\P_0,\P_1$ two probability measures on it. Let $(\mathcal{F}_n)_{n\in\N}$ be a sequence of $\sigma$-subalgebras of $\mathcal{M}$, and further assume that $\P_0|_{\mathcal{F}_n}\ll \P_1|_{\mathcal{F}_n}$ for all $n \in \N$. Then for any collection of events $(A_n^E)$ which decays uniformly exponentially with respect to $\P_1$ at rate $\eta$, we will also have uniform exponential decay with respect to $\P_0$ if
		\begin{equation}\label{eventsup}
			\eta_0:= \limsup_{n\rightarrow \infty} \left(\sup_E\frac{1}{n} \log\left( \left\|\chi_{A_n^E}H_n\right\|_\infty\right)\right) < \eta \end{equation}
		where\[H_n:= \frac{d(\P_0|_{\mathcal{F}_n})}{d(\P_1|_{\mathcal{F}_n})}\]
		and $\|\cdot\|_\infty$ denotes the $L^\infty(d\P_1)$ norm.
	\end{lem}
	
	\begin{proof}
		Under the assumption of (\ref{eventsup}) we have as before
		for any $\ve > 0$ some $N_0= N_0(\ve)$ such that
		\[ e^{-(\eta_0+\ve)n}\left\|\chi_{A_n^E}H_n\right\|_\infty \leq 1\]
		for all $n> N_0$ and $\alpha \in A$.
		By our assumption of uniform exponential decay with respect to $\P_1$, there is $N_1 \in \N$ such that
		\[ \P_1[A_n^E] \leq e^{-\eta n}\]
		for all $\alpha$ and $n >N_1$. Fixing $\ve < \eta-\eta_0$, and $N = \max(N_0,N_1)$, we get:
		\[\P_0[A_n^E] \leq e^{-(\eta-\eta_0-\ve)n}\]
		for $n > N$ and all $E$, 
		establishing uniform exponential decay of the family with respect to $\P_0$.
	\end{proof}
	
	Control of the Radon-Nikodym derivatives on the whole space gives a general result.
	\begin{thm}\label{slowgrowlemunifcor}
		If $(\Omega,\mathcal{M})$, $\mathcal{F}_n$, $\P_0$, $\P_1$ and $H_n$ are as above, and moreover we have
		\begin{equation}\label{zerolog} \lim_{n\rightarrow \infty} \frac{1}{n}\log\|H_n\| = 0
		\end{equation}
		then any family adapted to $\mathcal{F}_n$ exponentially decaying uniformly in $E$ with respect to $\P_1$ is exponentially decaying uniformly in $E$ with respect to $\P_0$.
	\end{thm}
	\begin{proof}
		For any family of events we have
		\[\sup_E\|\chi_{A_n^E}H_n\|_\infty  \leq \|H_n\|_\infty\]
		for all $n$. By monotonicity of $\log$, the required bound (\ref{eventsup}) holds as before.
		
	\end{proof}
	
	Along essentially the same lines we also have a uniformized version:
	
	\begin{thm}\label{unifslogro}
		If $(\Omega,\mathcal{M})$, $(\mathcal{F}_n)^k$, $\P_0$, $\P_1$ and $H_n^k$ are as above, and moreover we have
		\begin{equation}\label{zerologunif} \lim_{n\rightarrow \infty}  \frac{1}{n}\left[\sup_{k\in\Z}\log\|H_n^k\|_\infty \right] = 0
		\end{equation}
		then any family adapted to $\mathcal{F}_n$ exponentially decaying uniformly in $E$ and $k$ with respect to $\P_1$ is exponentially decaying uniformly in $E$ and $k$ with respect to $\P_0$.
	\end{thm}

	\begin{remm}
		These results have a natural analogue for other large deviation estimates. In particular if $r(n) \geq 0$ is monotone increasing and $r(n) \rightarrow \infty$, then for $\P_0$ and $\P_1$ satisfying the obvious analogue of (\ref{zerolog})
		\begin{equation}\label{zerolog4evr} \lim_{n\rightarrow \infty} \frac{1}{r(n)} \log\|H_n\|_\infty = 0
		\end{equation}
		any adapted sequence $A_n$ eventually satisfying $\P_1[A_n] \leq e^{-r(n)}$ also satisfies $\P_0[A_n] \leq e^{-(1-\ve)r(n)}$ eventually for any $\ve > 0$, and this can be made uniform over a parameter or over filtrations with the appropriate uniformity assumptions imposed on (\ref{zerolog4evr}).
	\end{remm}
	For our proof of localization, absent the uniformity in filtration coming from Theorem \ref{unifslogro}, we can get a weaker form of uniformity by examining arithmetic progressions within a filtration:
	\begin{cor}\label{arith}
		If $(\Omega,\mathcal{M})$, $\mathcal{F}_n$, $\P_0$ and $\P_1$ are as above and (\ref{zerolog}) holds, any family of events $\{A^E_n\}$ which is adapted with respect to a subfiltration $\{\mathcal{G}_n\}$ of the form $\mathcal{G}_{n} = \mathcal{F}_{kn+l}$ uniformly exponentially decays with respect to $\P_0$ if it does the same for $\P_1$.
	\end{cor}
	\begin{proof}
		It suffices to show that
		\[ \lim_{n\rightarrow 0} \frac{1}{n} \log\|H_{kn+l}\|= 0
		\]
		For large enough $n$ we have
		\[\frac{1}{n} \log\|H_{kn+l}\| \leq \frac{2k}{kn+l} \log\|H_{kn+l}\| \]
		with the right hand side term going to zero as a consequence of (\ref{zerolog}).
	\end{proof}
	We will later extract uniformity over finitely many (sub-)filtrations in the context of Schr\"odinger operators. Our work thus far now suffices to prove Theorem \ref{transmeth}, after recalling some definitions and a fundamental probabilistic result. Its uniformized (in filtration) variant Theorem \ref{transmethunif} will follow along the same lines.
	\begin{deffo}
		If $\Omega$ is a set, we call $\mathcal{A} \subset \mathcal{P}(\Omega)$ a $\pi$-system over $\Omega$ if
		\begin{enumerate}
			\item $\mathcal{A}$ is non-empty
			\item $\mathcal{A}$ is closed under finite intersections, i.e. for $A_1,A_2 \in \mathcal{A}$, we have $A_1 \cap A_2 \in \mathcal{A}$.
		\end{enumerate}
	\end{deffo}
	\begin{deffo}
		If $\Omega$ is a set, we call $\mathcal{Z} \subset \mathcal{P}(\Omega)$ a $\lambda$-system if
		\begin{enumerate}
			\item $\Omega \in \mathcal{Z}$
			\item $\mathcal{Z}$ is closed under complementation; for any $A \in \mathcal{Z}$, we have $A^C \in \mathcal{Z}$
			\item $\mathcal{Z}$ is closed under disjoint countable unions; if $(A_n)_{n\in\N}$ are pairwise disjoint and all in $\mathcal{Z}$, then $\cup A_n \in \mathcal{Z}$
		\end{enumerate}
	\end{deffo}
	A result of Dynkin, found in e.g. \cite{durrett_2019}, allows us to prove equality of measures in terms of these systems.
	\begin{thmcite}{(Dynkin)}
		If $\mathcal{A}$ is a $\pi$-system contained in some $\lambda$-system $\mathcal{Z}$,
		then $\sigma(\mathcal{A})$ is also contained in $\mathcal{Z}$, where $\sigma(\mathcal{A})$ is the $\sigma$-algebra generated by $\mathcal{A}$.
	\end{thmcite}
	We prove one more lemma, which, combined with Theorem \ref{slowgrowlemunifcor} implies Theorem \ref{transmeth}.
	\begin{lem} \label{duh}
		In the setting of  Theorem \ref{transmeth} we have for $N >0$:
		\[\|H_N\|_\infty = \prod_{n=-N}^N\|g_n\|_\infty\]
	\end{lem}
	\begin{proof}
		We prove this by showing that in fact we have the pointwise (a.s. with respect to $\P_1$) equality:
		\[ H_N(\omega) = \prod_{n=-N}^N g_n(V_n(\omega)) \]
		The right hand side of this clearly has norm less than or equal to $\prod_{n=-N}^N \|g_n\|_\infty$ by submultiplicativity of $\|\cdot\|_\infty$, and can be shown to have norm at least that more or less as a direct consequence of the product measure structure. This equality of functions (pointwise a.s.) will follow from showing  equality of their integrals on any $\mathcal{F}_n$ measurable set. We define
		\[ \tilde{H}_N := H_N(\omega) - \prod_{n=-N}^N g_n(V_n(\omega))
		\]
		so that equality of their integrals is the same as vanishing of the integral of $\tilde{H}_N$.
		
		We let $\mathcal{Z}_N$ be the family of sets $A \in \mathcal{F}_N$ such that
		\[\int_A \tilde{H}_N(\omega) \,d\P_1 = 0\]
		These families are closed under countable disjoint union as a consequence of the dominated convergence theorem. (Note that $\E_1[|H_N|] \leq 2$.) If we assume that $\Omega \in \mathcal{Z}_N$, we obtain:
		\[\int_A \tilde{H}_N(\omega)\,d\P_1 = -\int_{A^C} \tilde{H}_N(\omega)\,d\P_1 \]
		If we can show $\Omega \in \mathcal{Z}_n$, we will also have closure under complementation, and hence that $\mathcal{Z}_n$ are $\lambda$-systems.
		
		We let $\mathcal{A}_N$ denote the family of $\mathcal{F}_N$ measurable cylinder sets, i.e. $A = \prod_{n \in\Z} A_n$ with $A_n \in \mathcal{M}$. It is clear that $\mathcal{F}_N$ measurability requires that $A_n \in \{\varnothing, X\}$ for $n > N$ or $n < -N$. We will show $\mathcal{A}_N \subset \mathcal{Z}_N$. The result is trivial if $A_n = \varnothing$ for any $n$ as then the set has measure zero; we thus assume that $A_n = X$ for $n >N$ or $n <- N$. The result then follows by a computation, using the fact that $\P_1$ is a product measure:
		\begin{align*}
			\E_1[\chi_A H_N(\omega)] &=\P_0[A]\\
			&= \prod_{n = -N}^N \E_1[\chi_{A_n}g_n(V_n(\omega))]
		\end{align*}
		We note that $\Omega \in \mathcal{A}_N$ for all $N$, as clearly $\Omega = \prod_{n \in \Z} X$, hence $\Omega \in \mathcal{Z}_N$, and the $\mathcal{Z}_N$ are $\lambda$-systems. Moreover, $\mathcal{A}_N$ are $\pi$-systems; clearly \[\prod_{m \in Z} A_m \cap \prod_{m\in \Z} B_m = \prod_{m\in Z} (A_m \cap B_m)\]
		and measurability is preserved. Hence $\sigma(\mathcal{A}_N) \subset \mathcal{Z}_N$ by the $\pi$-$\lambda$ theorem. But $\sigma(\mathcal{A}_N)$ is precisely $\mathcal{F}_n$; the pre-images of the projections are cylinder sets, giving us the desired equality.
	\end{proof}
	
	Using a similar argument as in the previous lemma, we can prove an abstract result which has Theorem \ref{abscont} as an immediate consequence. First we recall a probabilistic result of Doob's, whence we will obtain the existence of $H_\infty = \lim_{N\rightarrow\infty} H_N$. The result can be found in e.g. \cite{durrett_2019}.
	\begin{thmcite} {(Doob)}
		Let $Y_n$ be a martingale on a probability space $(X,\mathcal{F},P)$ with respect to a filtration $\mathcal{F}_n$ such that $Y_n \geq 0$, and also let $\mathcal{F}_\infty = \sigma(\mathcal{F}_1,\dots)$. Then there is $Y_\infty$ which is $\mathcal{F}_\infty$ measurable such that $Y_n \rightarrow Y_\infty$ almost surely.
	\end{thmcite}
	Using this result, we can prove:
	\begin{lem}\label{abscontgen}
		Let $(\Omega,\mathcal{F})$, $\P_0$, $\P_1$ and $(\mathcal{F}_n)$ be as in the previous lemma and corollary, and let $\mathcal{F}_\infty = \sigma(\mathcal{F}_1,\mathcal{F}_2,\dots)$. If
		\begin{equation}\label{absbound}
			C:=\sup_n \log \|H_n\|  < \infty
		\end{equation}
		then $\P_0|_{\mathcal{F}_\infty} \ll \P_1|_{\mathcal{F}_\infty}$ and $\P_0[A] \leq e^C\cdot \P_1[A]$ for all $A \in \mathcal{F}$.
	\end{lem}
	\begin{proof}
		First we note that even in the absence of the moment condition, the stochastic process $H_n$ is in fact a martingale with respect to $(\mathcal{F}_n)$ and $\P_1$. Indeed, the martingale condition requires that $\E_1[H_{n+1}\,|\,\mathcal{F}_n] = H_n$. It follows from the definition that $H_n$ is $\mathcal{F}_n$ measurable, and because $(\mathcal{F}_n)$ is a filtration, $\mathcal{F}_n \subset \mathcal{F}_{n+1}$ so that for $A \in \mathcal{F}_n$, we have
		\begin{align*}
			\E_1[\chi_A H_n] &=\P_0[A]\\
			&= \E_1[\chi_A H_{n+1}]\\
			&= \E_1[\E_1[ \chi_A\cdot H_{n+1} | \mathcal{F}_n]]\\
			&= \E_1[\chi_A\cdot \E_1[H_{n+1}|\mathcal{F}_n]]
		\end{align*}
		with the very last equality a consequence of the $\mathcal{F}_n$ measurability of $A$. Hence $H_n$ is in fact $\E_1[H_{n+1}|\mathcal{F}_n]$.
		
		Non-negativity of the Radon-Nikodym derivatives is obvious, so there exists (up to almost everywhere equivalence) a $\P_1$ almost sure limit
		\[ H_\infty(\omega) := \lim_{n\rightarrow \infty} H_n(\omega) \]
		which is $\mathcal{F}_\infty$ measurable; moreover, under our condition (\ref{absbound}), satisfies $H_\infty \leq e^C$ almost surely. As a consequence of the dominated convergence theorem, $\E_1[H_\infty] = 1$. The theorem will follow immediately once we show that
		$\P_0[A] = \E_1[\chi_A\cdot H_\infty]$ for all $A \in \mathcal{F}_\infty$, i.e. $H_\infty$ is the Radon-Nikodym derivative $\frac{d(\P_0|_{\mathcal{F}_\infty})}{d(\P_1|_{\mathcal{F}_\infty})}$. Towards this end, we use another $\pi$-$\lambda$ argument. We let $\mathcal{A} = \cup_{n\in\N} \mathcal{F}_n$. Non-emptiness and closure under finite intersections are both obvious. We let $\mathcal{Z}$ be the collection of $A \in \mathcal{F}$ such that $\P_0[A] = \E_1[\chi_A \cdot H_\infty]$. As noted earlier, $\E_1[H_\infty] = 1$ by dominated convergence, and so $\Omega \in \mathcal{Z}$. Naturally then, $\E_1[\chi_{A^C}\cdot H_\infty] = 1 - \E_1[\chi_A\cdot H_\infty]$, so that $\mathcal{Z}$ is closed under complementation. Finally, using monotone convergence twice we obtain for any countable collection of disjoint $(A_n)_{n\in\N} \subset \mathcal{Z}$:
		\begin{align*}
			\P_0[\cup A_n] &= \sum_n\P_0[A_n] \\
			&= \sum_n \E_1[\chi_{A_n}\cdot H_\infty]\\
			&= \E_1[\chi_{\cup_n A_n}\cdot H_\infty]
		\end{align*}
		so that $\mathcal{Z}$ is a $\lambda$-system. It follows that $\sigma(\mathcal{A}) = \mathcal{F}_\infty$ is contained in $\mathcal{Z}$; this completes the argument.

	\end{proof}
	\begin{remm}
		The condition (\ref{absbound}) is a strict strengthening of (\ref{zerolog}) and relies only on the coarse data coming from the norms $\|H_n\|_\infty$. However, strictly speaking the weaker condition 
		\[ \E_1\left[\sup_{n\in \N}H_n(\omega)\right] < \infty \]
		suffices to give $\P_0 \ll \P_1$, as the bound (\ref{absbound}) is only used to justify application of the dominated convergence theorem. In this more general setting, a bound of the form $\P_0[A] \leq C \cdot \P_1[A]$ for $A \in \mathcal{F}_\infty$ will hold if and only if the limit of these Radon-Nikodym derivatives $H_\infty$ is essentially bounded. Further, we note that (\ref{absbound}) is sharp in the sense that no weaker condition formulated solely in terms of the asymptotics of $\|H_n\|_\infty$ can serve as a sufficient condition for absolute continuity.
	\end{remm}
	We note that Theorem \ref{abscont} is an immediate consequence of this result together with Lemma \ref{duh}. Theorem \ref{dynloc} is then immediate consequence of the work in \cite{SJZhu19} together with Theorem \ref{abscont}. While we omit the details, the work in \cite{Rangamani19} establishing Anderson localization for Jacobi operators can also be extended to non-i.i.d. potentials in the strongly converging regime using Theorem \ref{abscont}.
	
	\section{Schr\"odinger operator preliminaries}\label{schrodd}
	Having shown some general results, we can now introduce notions relevant to the analysis of random Schr\"odinger operators. We also give some remarks on the general strategy of the non-perturbative approach. More or less the entire section follows \cite{SJZhu19} and \cite{Rangamani19} closely. For an introduction to fundamental results in the theory of random Schr\"odinger operators, we recommend section 9 of \cite{cycon1987}.
	Throughout this section, we let $\Omega = \R^\Z$, $\mathcal{B}$ the $\sigma$-algebra generated by cylinder sets (with respect to the Borel $\sigma$-algebra on $\R$), and $V_n$ be the coordinate projections $V_n(\omega):= \omega_n$. Then to any $\omega \in \Omega$, there is a Schr\"odinger operator defined on $\ell^2(\Z)$ by
	\begin{equation} [H_\omega \psi] (n) = \psi(n+1) + \psi(n-1) + V_n(\omega) \psi(n)\tag{\ref{schrod}}
	\end{equation}
	\begin{deffo}
		A probabilistic family of Schr\"odinger operators is a Borel probability measure $\P$ on $(\Omega,\mathcal{B})$.
	\end{deffo}
	\begin{remm}
		In the case where $\P$ has support contained in $[-M,M]^\Z$ for some $M$, we can consider $H_\omega$ as a random variable valued in $B(\ell^2(\Z))$ which is the pushforward of $P$ under $\omega \mapsto H_\omega$; doing this explicitly in the general (unbounded) case is unwieldy, hence our identification with the probability on the space of potentials $\R^\Z$. In either setting, $H_\omega$ is weakly measurable in an appropriate sense.
	\end{remm}
	Throughout the rest of the paper, we freely identify a probabilistic family of  Schr\"odinger operators with the corresponding probability distribution $\P$ on $(\Omega,\mathcal{B})$. Moreover, while much of what we discuss in this section holds in considerable generality, we restrict ourselves to considering two types of $\P$. We will consider distributions $\P_1$ of the form $\P_1 = \mu^\Z$ for some Borel measure $\mu$ on $\R$, and $\P_0$ of the form $\P_0 = \otimes_{n\in\Z} g_n \mu$ for some $\mu$ The $\P_0$ under consideration will always satisfy at least (\ref{logmom}), if not one of the stronger conditions (\ref{logmomunif}) or (\ref{logsum}), so that we can understand $\P_0$ as in fact ``close'' in some sense to the $\P_1$ corresponding to the base distribution $\mu$. In relation to each other, we will call $\P_0$ approximate and $\P_1$ exact.
	
	Given fairly mild assumptions on $\mu$, there are many results regarding localization for the exact system $\P_1$. In \cite{ks81}, it was shown that if $\mu$ was absolutely continuous with bounded density, then Anderson localization and a form of dynamical localization hold. This result was extended to hold for singular measures in \cite{Carmona1987} by Carmona-Klein-Martinelli, who found Anderson localization to hold for arbitrary non-trivial $\mu$ satisfying the moment condition
	\begin{equation}\label{moment}
		\int |x|^\alpha d\mu(x) < \infty
	\end{equation}
	for some $\alpha > 0$.
	This paper used results regarding large deviations for the Lyapunov exponent, together with the multi-scale analysis developed in \cite{FS83}.
	
	At least in the context of operators studied in this work, two properties are of interest, both corresponding to localization in some sense of the mass of $e^{-itH_\omega}\psi$ as $t$ ranges over $\R$.
	\begin{deffo}
		We say an operator $H_\omega$ is Anderson localized if the spectrum is entirely pure point, and its eigenfunctions are exponentially decaying.
	\end{deffo}
	\begin{deffo}
		An operator $H_\omega$ is dynamically localized if there is some $A$ and $\mu> 0$ such that
		\[ \sup_{t\in\R} |\langle \delta_x , e^{-itH_\omega} \delta_y \rangle | \leq Ae^{-\mu|x-y|}\]
	\end{deffo}
	We introduce the latter definition for the sake of completeness; as we have mentioned earlier, dynamical localization can be obtained in the setting of Theorem \ref{dynloc} via our probabilistic method without any further work. The rest of our paper focuses on the proof of Anderson localization in the settings of Theorems \ref{specloc} and \ref{speclocunb}.
	\begin{deffo}
		$E \in \R$ is called a generalized eigenvalue of $H_\omega$ if there exists some non-zero $\psi \in \C^\Z$ with $|\psi(n)|$ polynomially bounded as $|n|\rightarrow \infty$ satisfying 
		\begin{equation}\label{eigeneq} H_\omega\psi = E\psi
		\end{equation}
		Such $\psi$ is then called a generalized eigenfunction.
	\end{deffo}
	The study of these suffices more or less entirely to show Anderson localization. This is a consequence of Sch'nol's theorem, which can be found in e.g. \cite{cycon1987}.
	\begin{thmcite} {(Sch'nol)}
		If all the generalized eigenfunctions of $H_\omega$ are exponentially decaying (i.e. $|\psi(n)|$ decays exponentially as $|n| \rightarrow \infty$), then $H_\omega$ has only pure point spectrum.
	\end{thmcite}
	We analyze the asymptotics of $|\psi(n)|$ using by working over truncations to finite boxes. We thus define for $a \leq b$ the operator $H_{[a,b],\omega}$ as the restriction of $H_\omega$ to $[a,b]$. This is $PH_\omega P$ for an appropriate choice of projection $P$, and we identify it with a $(b-a+1) \times (b-a+1)$ matrix.\\
	
	In relation to these truncations, we can define additional quantities:
	\[ P^E_{[a,b],\omega} = \det(H_{[a,b],\omega}-E)\]
	and
	\[ G^E_{[a,b],\omega}(x,y) = \langle \delta_x, (H_\omega-E)^{-1}\delta_y \rangle \]
	calling the latter quantity the Green's function. Note that the Green's function is only defined for $E \notin \sigma(H_{[a,b],\omega})$ and $x,y \in [a,b]$, and that moreover the inverse $(H-E)^{-1}$ is the matrix inverse, not the $B(\ell^2(\Z))$ inverse. The importance of these quantities comes from the well-known formulae
	\begin{equation}\label{eigen2green}
		\psi(x) = - G_{[a,b],\omega}^E(x,a)\psi(a-1) - G_{[a,b],\omega}^E(x,b)\psi(b+1),\qquad x \in [a,b]
	\end{equation}
	and
	\begin{equation}\label{green2det}
		|G_{[a,b],\omega}^E(x,y)|= \frac{|P_{[a,x-1],\omega}^E|\cdot |P_{[y+1,b],\omega}^E|}{|P_{[a,b],\omega}^E|},\qquad x \leq y
	\end{equation}
	The first formula reduces exponential decay of generalized eigenfunctions to exponential decay of the truncated Green's functions; the second allows us to study the asymptotics of these in terms of the asymptotics of determinants. Finally, we will study these using the transfer matrices. Given any solution $\psi \in \C^\Z$ to (\ref{eigeneq}), it satisfies for all $n \in \Z$:
	\begin{equation}\label{transmat}
		\begin{pmatrix} \psi(n+1) \\ \psi(n) \end{pmatrix} = \begin{pmatrix}E - V_n(\omega) & -1 \\ 1 & 0 \end{pmatrix} \begin{pmatrix}\psi(n) \\  \psi(n-1) \end{pmatrix} 
	\end{equation}
	The 2$\times$2 matrix in (\ref{transmat}) is called the one-step transfer matrix, and we denote it by $S_n^E(
	\omega)$. For $a \leq b$, we can define $S_{[a,b]}^E(\omega)$ as the unique matrix such that
	\begin{equation} \label{transmateq} \begin{pmatrix} \psi(a) \\ \psi(a-1) \end{pmatrix} = S_{[a,b]}^E(\omega) \begin{pmatrix}\psi(b+1) \\ \psi(b) \end{pmatrix} 
	\end{equation}
	The asymptotics of these matrices encode the asymptotics of generic solutions $\psi \in \C^\Z$ to $H_\omega \psi = E\psi$; the formula below allows us to control the asymptotics of the truncated determinants, and through those, the truncated Green's functions. Finally, we make use of the formula
	
	\begin{equation}
		S_{[a,b],\omega}^E = \begin{pmatrix} P_{[a,b],\omega}^E & -P_{[a+1,b],\omega}^E \\
			P_{[a,b-1],\omega}^E & -P_{[a+1,b-1],\omega}^E \end{pmatrix} 
	\end{equation} to estimate the determinants $P_{[a,b],\omega}^E$ by writing them as matrix elements:
	\begin{equation}\label{det2lyap}
		P_{[a,b],\omega}^E = \left\langle \begin{pmatrix} 1 \\ 0 \end{pmatrix}, S_{[a,b],\omega}^E \begin{pmatrix} 1 \\ 0 \end{pmatrix} \right \rangle
	\end{equation}
	whence, at least in the stationary case $\P_1 = \mu^\Z$, the Furstenberg theory gives us information regarding the asymptotics. If
	\begin{equation}
		\int |x|^\alpha d\mu(x) < \infty
		\tag{\ref{moment}}
	\end{equation}
	for some $\alpha > 0$, then Furstenberg's theorem and extensions thereof are applicable. In particular, the work of Furstenberg-Kesten shows that under weaker conditions than those proposed, the quantity
	\[ \gamma(E) := \lim_{n\rightarrow \infty} \frac{1}{n} E[\log\|S_{[1,n]}^E\| ]\]
	is defined, and almost surely we have
	\[ \frac{1}{n} \log\|S_{[1,n]}^E(\omega)\| \rightarrow \gamma(E)\]
	and the Furstenberg Theorem implies that as long as $\mu$ is non-trivial, $\gamma(E) >0$ for all $E$.
	Le Page showed in \cite{LePage} that under our conditions, we have exponential decay in the probability of large deviations, both for these quantities and for the magnitude of the corresponding matrix elements. In \cite{Tsay} it was found that this could be made uniform over a parameter varying over a compact set; in particular it is a straightforward application of work in \cite{Tsay}, also proved in \cite{Bucaj2019LocalizationFT} that:
	\begin{thmcite} {(Tsai, Bucaj et al)} Fixing $I \subset \R$ compact and $\ve  >0$, there exist $\eta > 0$ and $N$ such that for any $u,v \in \R^2$ with $\|u\| = \|v\| = 1$ and any $E \in \R$, we have
		\[ \P_1\left[\left|\frac{1}{n}\log|\langle u, S_{[1,n],\omega}^Ev\rangle| - \gamma(E)\right| > \ve\right] < e^{-\eta n}\]
		for $n > N$.
	\end{thmcite}
	Because $\P_1 = \mu^\Z$, in particular the map $T$ defined by $(T\omega)_n = \omega_{n+1}$ is measure preserving, so that the statistics of $S_{[a,b]}^E$ and those of $S_{[1,b-a+1]}^E$ are identical. Hence exponential large deviation estimates for  $P_{[a,b],\omega}^E$, uniform over a compact interval $I$, are a corollary of Tsay's theorem together with (\ref{det2lyap}).
	\begin{cor}\label{detLDE}
		Fix $I \subset \R$ be compact and $\ve  >0$. Then there are $\eta > 0$ and $N$ such that
		\begin{equation}\label{exactLDE} \P_1\left[\left|\frac{1}{b-a+1}\log|P_{[a,b],\omega}^E| - \gamma(E)\right|>\ve\right] < e^{-\eta(b-a+1)}
		\end{equation}
		for $b-a+1 > N$.
	\end{cor}
	
	In particular, this result gives in a crude sense
	\[ |G_{[a,b],\omega}^E(x,y)| \sim e^{-\gamma(E)|x-y|}\]
	A stronger, quantitative version of this result would imply Anderson localization as a consequence of (\ref{eigen2green}). This motivates the following definition:
	\begin{deffo}
		We say $x \in \Z$ is $(C,n,E,\omega)$ regular if
		\[|G_{[x-n,x+n],\omega}^E(x,x\pm n)| \leq e^{-Cn}\]
		and $x$ is $(C,n,E,\omega)$-singular if it is not regular for the same set of parameters.
	\end{deffo}
	In particular, there is a reformulation of Anderson localization in terms of this notion. Using formula (\ref{eigen2green}), Theorems \ref{specloc} and \ref{speclocunb} are straightforward consequences of the below:
	\begin{thm}\label{reg2loc}
		Under the assumptions of either Theorem \ref{specloc} or Theorem \ref{speclocunb}, there is $\Omega_0 \subset \Omega$ with $\P_0[\Omega_0] =1$ such that for every $\omega \in \Omega_0$ and $E \in \R$, there exist $N= N(E,\omega)$ and $C=C(E)$ such that for every $n > N$, both $2n$ and $2n+1$ are $(C,n,E,\omega)$-regular.
	\end{thm}
	All the relevant estimates in this section are established for a stationary $\P_1$, and so the next section derives the analogous large deviation estimate results for an appropriate $\P_0$, along with important consequences.

	\section{Adaptations for non-stationarity}\label{ldeland}
	Throughout this section, $\P_0$ is of the form specified in either Theorem \ref{specloc} or Theorem \ref{speclocunb}; we explicitly indicate when a result only applies in one context. By applying Theorem \ref{transmeth} to Corollary \ref{detLDE}, we can produce a $\P_0$ analogue of (\ref{exactLDE}):
	
	\begin{thm}\label{finiteLDE}
		Fixing $I \subset \R$ compact, for any $\ve > 0$ and $K \in \N$ there are $\eta = \eta(\ve,K) > 0$ and $N=N(\ve,K)$ such that
		\[\P_0\left[\left|\frac{1}{b-a +n(j_1+j_2)+1}\log| P_{[a+j_1n,b+j_2n],\omega}^E| - \gamma(E)\right| > \ve \right] <e^{-\eta (b+j_2n-a-j_1n+1)} \]
		for $ -K \leq a \leq b \leq K$, $-K\leq j_1 < j_2 \leq K$, $n>N$  and $E \in I$.
	\end{thm}
	\begin{proof}
		Fix $a,b \in \Z$ and $j_1,j_2 \in \N_0$. Let $l = \max\{|a|,|b|\}$ and $j = \max\{|j_1|,|j_2|\}$. Clearly the events
		\[ \{\omega\,:\, \left|\frac{1}{b-a+j_1n+j_2n+1}\log|P_{[a-j_1n,b+j_2n],\omega}^E| - \gamma(E)\right| > \ve\}\]
		are $\mathcal{F}_{jn+l}$ measurable. Moreover, they have $\P_1$ exponential large deviation estimates by Corollary \ref{detLDE}. By Corollary \ref{transmeth}, there are $\P_0$ exponential large deviation estimates, i.e. there are $\tilde{N}(\ve,a,b,j_1,j_2)$ and $\tilde{\eta}(\ve,a,b,j_1,j_2)$ such that
		\[\P_0\left[\left|\frac{1}{b-a+(j_1+j_2)n+1}\log| P_{[a-j_1n,b+j_2n]}^E| - \gamma(E)\right| > \ve\right] < e^{\tilde{\eta}(b-a+j_1n+j_2n+1)}\]
		for $n >\tilde{N}$.
		Taking $\eta$ to be the minimum of $\tilde{\eta}$ for $|a|,|b|,|j_1|,|j_2|$ all smaller than $K$, and $N$ to be the maximum of $\tilde{N}$ ranging over the same parameters, we obtain the desired $\eta$ and $N$.
	\end{proof}
	We proved this theorem which works even without the assumption of condition (\ref{zerologunif}) to illustrate what is used in the bounded case if said condition fails; however we have something stronger when (\ref{zerologunif}) holds, allowing the use of Theorem \ref{transmethunif}.
	
	\begin{thm}\label{unifLDE}
		In the context of Theorem \ref{speclocunb}, but not necessarily in the context of Theorem \ref{specloc}, fixing $I \subset \R$ compact, for any $\ve > 0$, there are $\eta = \eta(\ve) > 0$ and $N = N(\ve)$ such that for $b-a > N$,
		\[\P_0\left[\left|\frac{1}{b-a+1}\log |P^E_{[a,b],\omega}| - \gamma(E)\right|\right] < e^{-\eta(b-a+1)}\]
	\end{thm}
	\begin{proof}
		By stationarity of the shift map $(T\omega)_n = \omega_{n+1}$ in the $\P_1$ context, the large deviation estimates furnished by \cite{Tsay} are in fact uniform in the endpoints of the interval, i.e. uniform in filtration. The result is then immediate as a consequence of Theorem \ref{transmethunif}.
	\end{proof}
	For a certain application of the work of \cite{CS83}, we need to consider complexified energy, i.e. transfer matrices $S_{[a,b]}^z$ which characterize the solutions to $H_\omega\psi = z\psi$ in the sense of (\ref{transmateq}). For $z \in \C$ generally, the non-uniform version of Theorem \ref{finiteLDE} (i.e. for a single fixed energy $z \in \C$) also holds by applying Theorem \ref{transmeth} to the work of LePage in \cite{LePage}. (We believe the uniform result in \cite{Tsay} still applies in the context of complexified energy, but we are not sure of the details, and do not need it.) There is analogously $\gamma(z) = \lim_{n\rightarrow \infty} \frac{1}{n}\E_1[\log\|S_{[1,n]}^z\|]$ and \[\frac{1}{n}\log\|S_{[1,n],\omega}^z\| \rightarrow \gamma(z)\] $\P_1$ almost surely. The large deviation results imply that these $\P_1$ almost sure limits are also $\P_0$ almost sure limits.
	\begin{thm}\label{lyap4evr}
		For $\P_0$ and $\P_1$ as above, and any $z \in \C$, there is $\Omega_z$ such that $\P_0[\Omega_z] = 1$ and
		\[ \frac{1}{n}\log\|S_{[1,n],\omega}^z\| \rightarrow \gamma(z)\]
	\end{thm}
	\begin{proof}
		We note first that as a consequence of Theorem \ref{finiteLDE} or its non-uniform complex analogue, there are $\eta> 0$ and $N \in \N$ such that for $n > N$; we have
		\[\P_0\left[\left|\frac{1}{n}\log\|S_{[1,n],\omega}^z - \gamma(z)\right| > \ve\right] < e^{-\eta n}\]
		The eventual exponential decay implies summability, so that by Borel-Cantelli there exist for all $\ve > 0$ subsets $\Omega_{z,\ve} \subset \Omega$ with $\P_0[\Omega_{z,\ve}] = 1$ and for all $\omega \in \Omega_{z,\ve}$ we some $N=N(\omega)$ such that $n>N$ implies:
		\[ \left|\frac{1}{n}\log\|S_{[1,n],\omega}^z\| - \gamma(z)\right| < \ve \]
		We note that $\frac{1}{n}\log\|S_{[1,n],\omega}^z\| \rightarrow \gamma(z)$ precisely if $\omega \in \cap_{m\in\N} \Omega_{z,1/m}$. This intersection has probability 1 and so is the requisite $\Omega_z$.
	\end{proof}
	This last result is of some interest in its own right, demonstrating an ability to ``lift'' Lyapunov behavior across contexts; it is also necessary for the proof of localization in making it possible to apply certain results exploiting subharmonicity of $\gamma(z)$ in \cite{CS83}.
	
	We go through the details here, though the argument is fundamentally the same as in the original paper. For any fixed $z \in \C$ there is $\Omega_z \subset \Omega$ with $\P_0[\Omega_z] = 1$ such that for $\omega \in \Omega_z$ the quantity defined below
	\[ \gamma^+(\omega,z):= \limsup_{n\rightarrow +\infty} \frac{\log\|S_{[1,n],\omega}^z\|}{n}
	\]
	coincides with $\gamma(z)$. A result of Craig and Simon makes this in some sense uniform in the context of the exact Anderson model, and moreover says the same for the quantity
	\[ \gamma^-(\omega,z):= \limsup_{n\rightarrow +\infty} \frac{\log\|S_{[-n,-1],E,\omega}^z\|}{n}\]
	(In fact, the quantities $\gamma^+(\omega,z)$ and $\gamma^-(\omega,z)$ always coincide, and so we denote this quantity going forward by $\overline{\gamma}(\omega,z)$.)
	
	\begin{thmcite}{(Craig-Simon)}\label{cs1}
		Let $\mu$ be a distribution on $\R$ satisfying the condition \[\int \max\{0,\log |x|\} d\mu(x) < \infty\] and $\P_1 = \mu^\Z$. Then there exists $\Omega_1 \subset \Omega = \R^\Z$ with $\P_1[\Omega_1] = 1$ such that $\overline{\gamma}(\omega,E) \leq  \gamma(E)$ for all $\omega \in \Omega_1$ and $E \in \R$.
	\end{thmcite}
	\begin{remm}
		Craig and Simon originally proved the result for bounded Schr\"odinger operators, which in our context implies an absolute bound on $V_n(\omega)$. However, the proof straightforwardly generalizes to any family $H_\omega$ satisfying this mild moment condition, without which $\gamma(E)$ is not even guaranteed to exist.
	\end{remm}
	The key step to proving this result concerning $E \in \R$ was a theorem regarding the Lyapunov exponent in complexified energy:
	\begin{thmcite}
		{(Craig-Simon)} For an exact system whose parameters satisfy the assumptions in Theorem \ref{cs1}, $\gamma(z)$ is subharmonic, and for all $\omega \in \Omega$, the function $\overline{\gamma}(\omega,z)$ is submean.
	\end{thmcite}
	That $\overline{\gamma}(\omega,E)$ is still submean in the context of the approximate system is obvious; the change from the exact to approximate system amounts only to a change in probability measure and $\overline{\gamma}(\omega,E)$ is not an averaged quantity. On the other hand, subharmonicity of the $\P_0$ Lyapunov exponent $\tilde{\gamma}(z)$ is only obtained by showing its equality with the $\P_1$ Lyapunov exponent $\gamma(z)$; we recall some basic facts from the theory of subharmonic and submean functions:
	\begin{prop}
		If $f$ is submean and $E_0$ is fixed, then
		\[ f(E_0) \leq \lim_{r\rightarrow 0} \frac{1}{\pi r^2} \int_{|E-E_0|<r} f(E)\,d^2E\]
		and if $f$ is subharmonic, $E_0$ is fixed, then
		\[ f(E_0) = \lim_{r\rightarrow 0} \frac{1}{\pi r^2} \int_{|E-E_0|<r} f(E)\,d^2E\]
	\end{prop}
	Using this, we can now prove, more or less along the lines of the original argument in the original paper of Craig and Simon, the following:
	\begin{thm}
		There exists a subset $\Omega_0 \subset \Omega$ with $\P_0[\Omega_0] = 1$ such that for all $E \in \R$:
		\[ \overline{\gamma}(\omega,E) \leq \gamma(E)\]
	\end{thm}
	\begin{proof}
		Recall that $\overline{\gamma}(\omega,z)$ is submean, and $\gamma(z)$ subharmonic. Moreover, we've shown that for any fixed $z \in \C$ there is a $\P_0$ full measure subset $\Omega_z \subset\Omega$ such that $\gamma(\omega,z) = \gamma(z)$. By Fubini, there is a $\P_0$ probability 1 subset $\Omega_0$ such that for $\omega \in \Omega_0$, we have
		$\gamma(\omega,z) = \gamma(z)$ for a Lebesgue almost every $z \in \C$.
		
		Hence it suffices to show that for any $\omega$ such that $\overline{\gamma}(\omega,z) \leq \gamma(z)$ for Lebesgue almost all $z$, we in fact have it for all $z$. For such $\omega$, we have necessarily for any fixed $E \in \R$ and $r>0$
		\[ \int_{|z-E|<r} \overline{\gamma}(\omega,z)d^2z = \int_{|z-E|<r} \gamma(z)d^2z\]
		It follows immediately that
		\begin{align*}
			\overline{\gamma}(\omega,E) &\leq \lim_{r \rightarrow 0} \frac{1}{\pi r^2} \int_{|z-E|<r} \overline{\gamma}(\omega,z)d^2z\\
			&= \lim_{r\rightarrow 0} \frac{1}{\pi r^2}\int_{|z-E|<r} \gamma(z)d^2z\\
			&= \gamma(E)
		\end{align*}
	\end{proof}
	We reformulate this result quantitatively and in terms of the transfer matrices.
	\begin{cor}\label{CSfini}
		For $\P_0$-a.s. $\omega$ and any $\ve> 0$, there exists $N = N(\omega,\ve)$ such that for $n > N$, we have
		\[\max\left\{\|S_{[1,n],E,\omega}\|,\|S_{[-n,-1],E,\omega}^{-1}\| \right\} \leq e^{(\gamma(E)+\ve)n}\]
		and
		\[ \max\left\{\|S_{[n+1,2n],E,\omega}\|,\|S_{[2n+2,3n],E,\omega}^{-1}\| \right\} \leq e^{(\gamma(E)+\ve)n}\]
	\end{cor}
	\section{Main lemmas}\label{lemma4evr}
	We use now the large deviation results we've obtained to prove several technical lemmas, generally following \cite{Rangamani19}, making some simplifications due to our exclusive consideration of the Schr\"odinger case, rather than Jacobi operators as a whole. Throughout this section, we fix a compact interval of energies $[s,t] =:I \subset \R$. Of course, if we have localization almost surely with respect to any given compact interval, then (taking a countable intersection) we have almost sure localization for all energies. Continuity of $\gamma(E)$ was shown in \cite{Rangamani19}, using ideas from work in \cite{Furstenberg1983}., so that in particular $\inf_{E \in I} \gamma(E) > 0$. The proof proceeds by analyzing the sets where large deviations occur, so we define subsets of $I \times \Omega$:
	\[ B^+_{[a,b],\ve} = \left\{(E,\omega)\,:\,|P_{[a,b],\omega}^E| \geq e^{(\gamma(E)+\ve)(b-a+1)}\right\}\]
	and
	\[ B^-_{[a,b],\ve} = \left\{(E,\omega)\,:\,|P_{[a,b],\omega}^E| \leq e^{(\gamma(E)-\ve)(b-a+1)}\right\}\]
	and the sections
	\[ B_{[a,b],E,\ve}^\pm = \{\omega\,:\, (E,\omega) \in B^\pm_{[a,b],\ve}\}\]
	and
	\[ B_{[a,b],\omega,\ve}^\pm = \{E\,:\, (E,\omega) \in B_{[a,b],\ve}^\pm\}\]
	Moreover, we let $E_{j,[a,b],
		\omega}$ denote the $b-a+1$ eigenvalues (with multiplicity) of $H_{[a,b],\omega}$. An immediate consequence of (\ref{green2det}) is that:
	\begin{lem} \label{overunder}
		If $n \geq 2$ and $0 < \ve < \nu_I/8$, and $x$ is $(\gamma(E) - 8\ve, E, n, \omega)$-singular, then
		\[ E \in B^-_{[x-n,x+n],\ve}\cup B^+_{[x-n,x],\ve} \cup B^+_{[x,x+n],\ve}\]
	\end{lem}
	Our work in lifting LDEs and the existence of these for the exact case will allow us to prove technical lemmas like those in \cite{SJZhu19} and \cite{Rangamani19}. In particular, Theorem \ref{finiteLDE} suffices for most of these. We take $\eta_0 = \eta(\ve_0,4)$ where, $\eta$ is the large deviation parameter from Theorem \ref{finiteLDE}. 
	We now proceed through some technical lemmas, proved in either \cite{SJZhu19} or \cite{Rangamani19} for $\P_1$, and comment on any differences arising in the proof in a non-stationary context.
	\begin{lem}\label{indlem}
		Let $0 < \delta_0 < \eta_0$. There is $\Omega_1 \subset \Omega$ such that $\P_0[\Omega_1] = 1$ and for all $\omega \in \Omega_1$ there is $N_1 = N_1(\omega)$ such that for all $n > N_1$, we have
		\[\max\{|B^-_{[n+1,3n+1],\ve_0,\omega}|,|B^-_{[-n,n],\ve_0,\omega}|\}  \leq e^{-(\eta_0-\delta_0)(2n+1)}\]
	\end{lem}
	This is done by a Borel-Cantelli argument in \cite{Rangamani19} for $\P_1$ which uses large deviation estimates and carries over with no modification to our $\P_0$ context as a consequence of Theorem \ref{finiteLDE}.
	
	\begin{lem}\label{edgelem}
		For any $\ve> 0$ and $p > 4/\eta_\ve$, where $\eta_\ve$ is the large deviation parameter furnished by Theorem \ref{unifLDE}, there is $\Omega_2=\Omega_2(\ve,p)\subset \Omega$ with full probability such that for every $\omega \in \Omega_2$ there is $N_2 = N_2(\omega)$ so that for $n>N_2$, any $y_1,y_2$ satisfying $-n\leq y_1 \leq y_2 \leq n$ and $|-n-y_1| > p\log n$ and $|n-y_2|> p\log n$, we have
		\[ E_{j,[n+1,3n+1],\omega} \not\in B_{[-n,y_1],\ve,\omega} \cup B_{[y_2,n],\ve,\omega}\cup B_{[-n,n],\ve,\omega}\]
		for all $j \in [1,b-a+1]$.
		
	\end{lem}
	\begin{proof}
		Following \cite{SJZhu19}, we analyze the events
		\[ A_n = \{\omega\,:\, \exists j \in [1,2n+1]; E_{j,[n+1,3n+1],\omega} \in B_{[-n,y_1],\ve,\omega} \cup B_{[y_2,n],\ve,\omega}\}\]
		and, leveraging independence and union bounds together with large deviation estimates coming from Theorem \ref{unifLDE}, obtain
		\begin{equation}\P_0[A_n] \leq 2(2n+1)^3e^{-\eta_\ve p\log n +2} \label{estim}
		\end{equation}
		for sufficiently large $n$. Because $p\eta_\ve > 4$, this is summable, whence the result follows.
	\end{proof}
	
	This result has an analogue in the context of Theorem \ref{specloc} which can be proven using only Theorem \ref{finiteLDE} rather than requiring Theorem \ref{unifLDE}.
	\begin{lem}\label{edgelemvar}
		For any $\ve> 0$ and $L > 1$, there is $\Omega_2=\Omega_2(\ve,L)\subset \Omega$ with full probability such that for every $\omega \in \Omega_2$ there is $N_2 = N_2(\omega)$ so that for $n>N_2$, any $y_1,y_2$ satisfying $-n\leq y_1 \leq y_2 \leq n$ and $|-n-y_1| > \frac{n}{L}$ and $|n-y_2|> \frac{n}{L}$, we have
		\[ E_{j,[n+1,3n+1],\omega} \not\in B_{[-n,y_1],\ve,\omega} \cup B_{[y_2,n],\ve,\omega} \cup B_{[-n,n],\ve,\omega}\]
		for all $j \in [1,b-a+1]$.
		
	\end{lem}
	
	Its proof is more or less the same as that of Lemma \ref{edgelem}, but replacing $p \log n$ with $\frac{n}{L}$ gives (in the original stationary context) a stronger estimate than (\ref{estim}):
	\begin{equation}\label{estimstrong}
		\P_1[A_n] \leq 2(2n+1)^3 e^{-\tilde{\eta}\frac{n}{L} +2}
	\end{equation}
	where $\tilde{\eta}$ is the $\P_1$ large deviation parameter. For any $\ve> 0$ and sufficiently large $n$, we then obtain:
	\[ \P_1[A_n] \leq e^{-(\frac{\tilde{\eta}}{L}-\ve)n}\]
	In particular then, instead of carrying over exponential estimates beforehand to ultimately produce the summable but subexponential estimates of (\ref{estim}), we derive exponential estimates of (\ref{estimstrong}) in the stationary context and then carry them over using Theorem \ref{finiteLDE}, obtaining
	\begin{equation}
		\P_0[A_n] \leq e^{-\eta n}
	\end{equation} (Even Theorem \ref{finiteLDE} is in some sense more than we need; another direct application of Theorem \ref{transmeth} suffices.) This  result is also true in the unbounded context of Theorem \ref{speclocunb}, but in carrying out the localization proof for the unbounded case, Lemma \ref{edgelem} turns out to be necessary. This is because the next pair of results concerning determinants corresponding to the edge of a box is weaker than its bounded context analogue.
	\begin{lem}\label{edgelemsmall}
		For fixed $r > 1$ and $p> 0$, for almost all $\omega$, there exists $N=N(\omega)$ such that for $n > N$ and $m \in [-n,n]$ with $|-n-m| \leq p \log n$ or $|n-m| \leq p \log n$
		\[ |V_m(\omega)| \leq n^{r/\alpha}\]
		where $\alpha > 0$ is such that $\sup_n \left[\int |x|^\alpha g_n(x)\,d\mu(x)\right] < \infty$.
	\end{lem}
	\begin{proof}
		We set $C:=\sup_n \left[\int |x|^\alpha g_n(x)d\mu(x)\right]$. By Chebyshev:
		\[\P_0[V_m(\omega) \geq n^{r/\alpha}] \leq \frac{C}{n^r}\]
		for any $m \in [-n,n]$. Hence the probability that there exists some $m$ with $|V_m(\omega)|$ exceeding $n^{r/\alpha}$ and also $|-n-m| \leq p \log n$ or $|n-m| \leq p\log n$ is bounded by (using a union bound): \[\frac{2C}{n^r}\left(1+2p\log n\right)\]
		which is summable, whence the result follows from Borel-Cantelli.
	\end{proof}
	A crucial corollary is the following:
	\begin{cor}\label{edgelemsmallcor}
		Fixing $I \subset \R$ compact, for $p > 0$ and $r > 1$, there is a probability 1 subset $\Omega_3 = \Omega_3(p,r) \subset \Omega$ such that for $\omega \in \Omega_3$, there is $N = N(\omega) \geq 3$ such that if $n > N$ and $|-n-y| > p\log n$, then 
		\begin{equation}
			|P_{[-n,y],\omega}^E| \leq e^{4pr\alpha^{-1} (\log n)^2}
		\end{equation}
		The same result also holds for $|n-y| < p \log n$ and $|P_{[y,n],\omega}^E|$ substituted in appropriately.
	\end{cor}
	\begin{proof}
		The proof of Lemma \ref{edgelemsmall} only relied on the existence of a uniform (in $m$) bound on the $\alpha$-th moments of $|V_m(\omega)|$. Clearly if $E$ varies over compact $I \subset \R$, there is similarly a uniform (in $m$) bound on the $\alpha$-th moments of $\sup_{E \in I} |V_m(\omega)-E|$. Hence, replacing $C$ in our proof with larger $\tilde{C}$ if necessary, we obtain the necessary result. Because of the tri-diagonal nature of $H_{[-n,y],\omega}$, we have \[P_{[-n,y],\omega}^E = \prod_{m=-n}^y (V_m(\omega) - E)\]
		whence we obtain
		\[ |P^E_{[-n,y],\omega}| \leq n^{\frac{2r}{\alpha}(p \log n +1)(\log n)} \]
		for $n$ large. Because $1 < p \log n$ for large $n$ the result follows immediately.
	\end{proof}
	In the bounded context, there is an obvious analogue of this result; there exists $M$ such that $|V_n(\omega)| \leq M$ for all $n$, and so (taking $\tilde{M}$ slightly larger if necessary to account for varying $E$ over $I$) $|P_{-[n,y],\omega}^E| \leq \tilde{M}^{b-a+1}$, without even making any restrictions to edges. Hence in particular:
	\begin{prop}\label{edgelemsmallcorvar}
		In the setting of Theorem (\ref{specloc}) with $\text{supp } \P_0 \subset [-M,M]^\Z$, if we fix $I \subset \R$ compact and $L>1$, then for all $\omega \in \text{supp }\P_0$, there is $N = N(\omega)$ such that if $n > N$ and $|-n-y| > \frac{n}{L}$, then 
		\begin{equation}
			|P_{[-n,y],\omega}^E| \leq \tilde{M}^{\frac{n}{L}}
		\end{equation}
		where $\tilde{M} = \sup_{E\in I} |M-E|$.
		The same result also holds for $|n-y| < \frac{n}{L}$ and $|P_{[y,n],\omega}^E|$ substituted in appropriately.
	\end{prop}
	The availability of such a bound is what allows the use of Lemma \ref{edgelemvar} instead of \ref{edgelem}. Ultimately, the problem of using Lemma \ref{edgelemvar} with our weaker control of $|V_m(\omega)|$ comes down to the superexponential growth of $n^{r\frac{n}{K}}$, whereas $n^{pr\log n} = e^{pr(\log n)^2}$ grows subexponentially and $\tilde{M}^{n/L}$ grows exponentially. In the following proof of Theorem \ref{reg2loc} under the assumptions of \ref{speclocunb}, we note that replacing $\eta_0 := \eta(\ve_0)$ and $\eta_\ve := \eta(\ve)$ from Theorem \ref{unifLDE} with $\tilde{\eta}_0:=\eta(\ve_0,4)$ and $\tilde{\eta}_\ve := \eta(\ve,4)$ from Theorem \ref{finiteLDE}, Lemma \ref{edgelem} with Lemma \ref{edgelemvar}, and $e^{pr(\log n)^2}$ with $\tilde{M}^{n/L}$ where $L$ is chosen to be sufficiently large (specifically $L > \frac{3\log \tilde{M}}{\ve_0 - \delta_0}$ in terms of parameters introduced in Section \ref{fini4evr}), we obtain a proof of Theorem \ref{reg2loc} under the assumptions of Theorem \ref{specloc}.
	
	\section{Proof of Theorem \ref{reg2loc}}\label{fini4evr}
	Having established these technical results in our non-stationary case, we are ready to prove $\P_0$ almost sure eventual $(\gamma(E)- 8\ve,n,E,\omega)$-regularity of $2n+1$, where $\ve < \nu_I/8$. The proof for $2n$, $-2n$ and $-2n-1$ follow by a nearly identical argument. (In the i.i.d. potential setting, symmetry considerations obviate any need to consider the negative cases even in passing; the absence of stationarity here technically forces one to run the argument in the negative direction as well, though no unique technical difficulties arise.)
	
	This follows \cite{Rangamani19}, which made some necessary adjustments to account for unboundedness. We note that many technical details there are absent here because they arise from considering the Jacobi case rather than only the Schr\"odinger case; more or less of all of these reduce to showing bounds on the growth of hopping terms, which are uniformly 1 in the Schr\"odinger case. We believe that these additional considerations don't prevent the argument from going through in the more general Jacobi case, but have not gone through the details.
	
	Though the argument more or less follows \cite{SJZhu19} and \cite{Rangamani19}, the former discussing bounded Schr\"odinger operators and the latter unbounded Jacobi operators, both of which were inspired by \cite{Jit99}, we say a few words about the general strategy. A particular site $n$ being $(\gamma(E) - \ve,n,E,\omega)$-singular forces a ``resonance'' of a sort; for sufficiently large $n$ Corollary \ref{CSfini} is used to establish the existence of eigenvalues $E_i$ and $E_j$ for truncations to $[-n,n]$ and $[n+1,3n+1]$ which give ``sub''-deviations from the expected Lyapunov behavior, either on the whole intervals $[-n,n]$ and $[n+1,3n+1]$ or on subintervals. Using lemmas \ref{edgelem} and \ref{edgelemsmall}, we derive three inequalities for different cases, at least one holding for each instance of singularity. All of these inequalities fail for large enough $n$; it follows that there cannot be infinitely many singular points with respect to any set of parameters satisfying our assumptions, giving localization with respect to the compact interval $I$.
	\begin{proof}[Proof of Theorem \ref{reg2loc}]
		We fix $0 < \ve_0 < \nu_I/8$. We then choose $\eta_0$ a large deviation parameter satisfying the conclusion of Theorem \ref{unifLDE} for $\ve_0$. Then given these parameters we fix $0 < \delta_0 < \eta_0$ and $0 < \ve < \min\{(\eta_0-\delta_0)/3, \ve_0\}$. Given these parameters, we take $\tilde{\Omega}$ with probability 1 such that the conclusions of Corollary \ref{CSfini} and Lemmas \ref{edgelem}, \ref{edgelemsmall} hold.
		
		We let $\omega \in \tilde{\Omega}$ and $E \in I$ be a generalized eigenvalue of $H_\omega$. We further let $\psi$ be the associated generalized eigenfunction. At least one of $\psi(0),\psi(1)$ is non-zero; we assume without loss of generality that $\psi(0) \neq 0$. There is $N(\omega)$  satisfying the conclusions of Corollary \ref{CSfini} and Lemmas \ref{edgelem}, \ref{edgelemsmall} and so that furthermore for $n > N$ we have that $0$ is $(\gamma(E)-8\ve_0, n, E, 
		\omega)$-singular. (This singularity is a consequence of $\psi(0) \neq 0$ together with (\ref{eigen2green}) and the polynomial bound on growth of $|\psi(n)|$.)
		
		We suppose (towards a contradiction) that for infinitely many $n > N$, $2n+1$ is also $(\gamma(E) - 8\ve_0, n, E, \omega)$-singular. By Lemma \ref{overunder} and Corollary \ref{CSfini}, we have that $E \in B^-_{[n+1,3n+1],\ve_0,\omega}$. There is then for any fixed $n >N$ some $E_j$, an eigenvalue of $H_{\omega,[n+1,3n+1]}$ such that $E_j$ lies in a slightly larger band of energies $\tilde{I} := [s-1,t+1]$ (recall that $I = [s,t])$, and $|E - E_j| \leq e^{-(\eta_0 - \delta_0)(2n+1)}$. Were this not the case, then either all eigenvalues would lie to one side of $\tilde{I}$, or some would lie on each side. In full detail, we show that the latter case is impossible, with a proof that generalizes straightforwardly to the former. We let $E_{j^-}$ be the largest eigenvalue to the left of $I$, and $E_{j^+}$ the smallest to the right. 
		
		All $E_j$ are the real zeroes of $P_{[n+1,3n+1],E,\omega}$, which is a polynomial in $E$ of degree $2n+1$. Then $P_{[a,b],\omega}^E$ is monotone on one of $[E_{j^-},E]$ or $[E,E_{j^+}]$. Then in particular we have
		\[ 1 \leq \min\{|E_{j^-} - E|,|E_{j^+}-E|\} \leq m(B^-_{[a,b],\omega}) \leq e^{-(\eta_0-\delta_0)(2n+1)} < 1\]
		whence we conclude  it is impossible that there are no eigenvalues in $\tilde{I}$, and so there is some $E_j \in \tilde{I}$. Our proof in particular showed that one such $E_j$ satisfies $|E-E_j| \leq e^{-(\eta_0 - \delta_0)(2n+1)}$. We can repeat this argument to produce an eigenvalue $E_i$ of $H_{\omega,[-n,n]} \in B^-_{[-n,n],\ve,\omega}$ such that $|E_i - E| \leq e^{-(\eta_0-\delta_0)(2n+1)}$. Then $|E_i - E_j| \leq 2e^{-(\eta_0-\delta_0)(2n+1)}$. By Lemma \ref{edgelem} we have that in particular that $E_j \notin B_{[-n,n],\omega,\ve}$. However, because $E_i$ is an eigenvalue of $H_{[-n,n],\omega}$, we obtain:
		\[ \|G_{[-n,n],E_j,\omega}\| \geq \frac{1}{2}e^{(\eta_0-\delta_0)(2n+1)}\]
		and by equivalence of norms there are $y_1,y_2 \in [-n,n]$ such that $y_1 \leq y_2$ and
		\begin{equation}\label{biggreen} |G_{[-n,n],E_j,\omega}(y_1,y_2)| \geq \frac{1}{2\sqrt{2n+1}}e^{(\eta_0-\delta_0)(2n+1)}
		\end{equation}
		These two facts together will yield precisely the sought contradiction. Because $E_j \notin B_{[-n,n],\ve,\omega}$, we obtain
		\[ |P_{[-n,n],\omega}| \geq e^{(\gamma(E_j) - \ve)(2n+1)}\] which we can combine with (\ref{green2det}) to obtain
		\begin{equation}\label{bigdet} |P_{[-n,y_1-1],\omega}^{E_j}|\cdot|P_{[y_2+1,n],E_j,\omega}^{E_j}| \geq \frac{e^{(\eta_0-\delta_0 + \gamma(E_j)-\ve)(2n+1)}}{2\sqrt{2n+1}}
		\end{equation}
		Because $y_1 < y_2$, our analysis can be split into (essentially) three cases; $|-n-y_1| \geq p\log n$ and $|n-y_2| \geq p\log n$, $|-n-y_1| \geq \log n$ but $|n-y_2| < \log n$,  $|-n-y_1| < \log n$ and $|n-y_2| < p\log n$. (A fourth case mirroring the second also exists; we omit any explicit consideration as the argument is the same.) In the first case, it is an immediate consequence of Lemma \ref{edgelem} that (\ref{bigdet}) yields
		\begin{equation}
			e^{(\gamma(E_j)+\ve)(2n+1)} \geq \frac{1}{2\sqrt{2n+1}} e^{(\eta_0 - \delta_0 +\gamma(E_j)-\ve)(2n+1)}
		\end{equation}
		By our choice of $\ve < (\delta_0 - \eta_0)/3$, this cannot hold for arbitrarily large $n$.
		
		For the second and third cases, we use Lemma \ref{edgelemsmall} to bound the term ``close'' to the edge.  Hence, for the second case (\ref{bigdet}) yields
		\[ e^{3pr\alpha^{-1}(\log n)^2} e^{(\gamma(E_j)+\ve)(2n+1)} \geq  \frac{1}{2\sqrt{2n+1}} e^{(\eta_0 - \delta_0 +\gamma(E_j)-\ve)(2n+1)} \]
		which also cannot hold for arbitrarily large $n$. In the third case (\ref{bigdet}) yields
		\[ e^{6pr\alpha^{-1}(\log n)^2} \geq  \frac{1}{2\sqrt{2n+1}} e^{(\eta_0 - \delta_0 +\gamma(E_j)-\ve)(2n+1)}\]
		which again, cannot hold for arbitrarily large $n$. Hence it is impossible that that there are infinitely many $n$ such that $2n+1$ is $(\gamma(E)-\ve, n, E,\omega)$-singular.
	\end{proof}
	
	\section*{Acknowledgments}
	We would like to thank Lana Jitomirskaya for suggesting the problem to us and for many helpful discussions, Rui Han for reading early drafts closely and providing helpful comments, and Anton Gorodetski and Victor Kleptsyn also for reading an early draft and making suggestions. This work was partially supported by NSF DMS-2052899,
	DMS-2155211, and Simons 896624.
	
	\printbibliography
	
\end{document}